\documentclass[letterpaper, 10 pt, conference]{ieeeconf}  

\IEEEoverridecommandlockouts                              

\usepackage{graphicx}
\usepackage{amsmath,amssymb,amscd,bbm}
\usepackage{bm}
\usepackage{float}
\usepackage{threeparttable}
\usepackage{color}
\usepackage{geometry}
\usepackage{booktabs}
\usepackage{mathtools}
\usepackage[compress]{cite}
\usepackage{makecell,subfigure}
\newtheorem{theorem}{Theorem}
\newtheorem{proposition}{Proposition}

\newtheorem{definition}{Definition}

\usepackage{hyperref}
\usepackage{algorithm}
\usepackage{algpseudocode}
\title{\LARGE \bf
Maximal quantum leakage: operational interpretation and quantum channel analysis
}

\author{Shuixin Xiao, Zijia Zhao, Jingge Zhu and Farhad Farokhi 
\thanks{This research was supported by the Faculty of Engineering and Information Technology at the University of Melbourne.}
\thanks{The authors are with the Department of Electrical and Electronic Engineering, The University of Melbourne, Parkville, VIC 3010, Australia (e-mail: shuixin.xiao@unimelb.edu.au; zijia.zhao@student.unimelb.edu.au; jingge.zhu@unimelb.edu.au; farhad.farokhi@unimelb.edu.au).}
}

\begin{document}

\maketitle
\thispagestyle{empty}
\pagestyle{empty}


\begin{abstract}
Maximal quantum leakage quantifies privacy against adversaries with arbitrary intentions. In this work, we prove that computing this leakage is equivalent to minimum-error quantum state discrimination with equal priors. This establishes a computable operational interpretation, addressing the previous difficulty in computing maximal quantum leakage. We further analyze the impact of collective measurements on multiple copies of a state, demonstrating that leakage increases monotonically with the number of copies, which leads to explicitly characterizing the maximal leakage in the asymptotic limit. Extending this framework to quantum channels, we develop an iterative algorithm for the jointly designing of input states and measurements. Numerical examples involving collective measurements and the maximal channel leakage demonstrate our theoretical findings.
\end{abstract}

\section{Introduction}
Advantages of quantum computing~\cite{PhysRevLett.85.441}, quantum simulation~\cite{RevModPhys.86.153,QIN20241}, and quantum sensing~\cite{PhysRevResearch.7.L012049,PhysRevA.103.042418} have motivated extensive research in this domain. However, the growing commercial accessibility of quantum technologies raises critical privacy concerns, particularly for sensitive datasets. This necessitates rigorous frameworks to quantify information leakage and design secure, privacy-preserving algorithms~\cite{10115324,8049724,11006077}.

In the classical setting, information theory has paved the way for development of unsuitable metrics for security analysis~\cite{8943950}. This motivated development of maximal leakage as a robust measure tailored specifically for security~\cite{8943950,9457602}. Recent efforts to extend this concept to the quantum domain have led to the proposal of maximal quantum leakage~\cite{Farhad2024}. This measure satisfies the fundamental axiomatic properties required for a rigorous security framework, including positivity, independence, and the post-processing inequality. Building on this foundation, gentle quantum leakage was proposed in~\cite{10886009}, and optimal quantum encoding was investigated in~\cite{farokhi2024optimal}. However, a significant limitation of the formulation in~\cite{Farhad2024} is the absence of a closed-form expression for the leakage. Consequently, computing maximal leakage generally relies on iterative algorithms similar to those used for accessible information~\cite{PhysRevA.71.054303}, which do not possess convergence guarantees to the global optimum. To address this, two alternative measures, Barycentric and pairwise quantum R\'{e}nyi leakages, were proposed in~\cite{Farokhi20242} as upper bounds on maximal quantum leakage, which can be loose generally.

To address this computational challenge and provide a clear operational interpretation, we establish a fundamental link to quantum state discrimination. This framework inherently characterizes distinguishability of quantum states and governs extraction of classical information from quantum systems~\cite{Barnett09,Bae2015}. Consequently, quantum state discrimination is central to information security as it defines the ultimate limit on the information an adversary can retrieve. The rigorous study of optimal discrimination strategies and measurement statistics provides a firm mathematical foundation essential for analyzing security limits~\cite{1055351,Helstrom1969}. By mapping maximal leakage to this well-established framework, we can leverage powerful analytical tools where the optimal strategy is determined by specific figures of merit, such as minimizing the average error or maximizing detection confidence.
\vspace{-1pt}

The main contributions of this work are summarized as follows. We first provide a rigorous proof establishing that maximal quantum leakage is equivalent to the minimum-error quantum state discrimination problem under equal priors, which yields a computable operational interpretation of the leakage. We extend this analysis to collective measurements performed on multiple copies of a quantum state. Collective measurements have been widely investigated in quantum tomography, where they have been shown to extract significantly more information than individual measurement strategies~\cite{Hou2018,PhysRevLett.120.030404,7956181,11158864,xiao2025generalized}. We establish that maximal quantum leakage increases monotonically with the number of available copies and characterize its asymptotic limit.
Furthermore, we generalize this framework to quantify information leakage within quantum channels. Quantum channels represent a fundamental quantum resource and have been widely discussed in the literature~\cite{10384297,xiaoqpt,10383498}.
To address non-convexity of maximal channel leakage problem, we develop an iterative algorithm alternating between optimization of input states and measurement operators. Because the subproblem in each iteration is convex, our algorithm is guaranteed to converge.

The rest of this paper is structured as follows. Section~\ref{sec2} introduces the necessary preliminaries concerning quantum leakage and state discrimination. Section~\ref{sec3} establishes the formal relationship between leakage and state discrimination and includes an analysis of collective measurements. Section~\ref{sec4} presents the analytical framework for quantifying leakage in quantum channels. Numerical results and simulations are discussed in Section~\ref{sec5}, and Section~\ref{sec6} provides concluding remarks.

\section{Preliminary}\label{sec2}
\subsection{Random variables, information measures, quantum states and information}

Random variables are denoted by capital Roman letters, such as $X \in \mathcal{X}$ and $Y \in \mathcal{Y}$. A random variable $X$ is discrete if its alphabet $\mathcal{X}$ is finite. Any discrete random variable $X$ is fully characterized by its probability mass function (PMF) $P_X(x) := {P}\{X = x\}>0$. Throughout this paper, $\log$ denotes the natural logarithm.
The infinite Sibson mutual information \cite{8804205,11078297} is
\begin{equation}
    I_{\infty}^{S}(X; Y) = \log \Big(\sum_{y \in \mathcal{Y}} \max_{x \in \mathcal{X}} P_{Y|X}(y|x)\Big).
\end{equation}
Further details regarding Sibson mutual information can be found in~\cite{8804205,11078297,11240362,10619672}.

Let $\mathcal{H}$ denote a finite dimensional Hilbert space. A quantum state is represented by a density operator $\rho$ belonging to the set $\mathcal{S}(\mathcal{H})$ which consists of positive semidefinite operators from $\mathcal{H}$ to itself with unit trace. This work focuses on states within a $d$ dimensional quantum system. Quantum measurements are characterized by positive operator valued measures (POVMs) defined as a set of operators $\{F_i\}$ satisfying $F_i \ge 0$ and $\sum_i F_i = I$. According to the Born rule, the probability of obtaining measurement outcome $i$ for a state $\rho$ is given by $\operatorname{Tr}(\rho F_i)$. Furthermore, a quantum channel $\mathcal{E}$ is defined as a completely positive and trace preserving map that transforms states from $\mathcal{S}(\mathcal{H}_A)$ to $\mathcal{S}(\mathcal{H}_B)$.

\subsection{Quantum state discrimination}
In quantum state discrimination, there are many different cost functions, including minimum-error discrimination, unambiguous state discrimination, and maximum-confidence discrimination. The minimum-error discrimination problem considers a scenario where a quantum system is prepared in one of $N$ possible states $\{\rho_i\}_{i=1}^N$ with corresponding prior probabilities $\{q_i\}_{i=1}^N$. The objective is to perform a measurement on the system to correctly identify which specific state was prepared. We aim to find the optimal measurement, described by a set of operators $\{M_i\}$ where yielding outcome $i$ dictates our guess that the state is $\rho_i$, that maximizes the overall probability of a successful identification. This maximum success probability is given by~\cite{Barnett09,Bae2015}
\begin{equation}
\begin{aligned}
    P_{\text{guess}} = \max_{\{M_i\}} \; & \sum_{i=1}^N q_i \text{Tr}(M_i \rho_i) \\
    \text{subject to} \;& \sum_{i=1}^N M_i = I, \quad M_i \geq 0, \,\, \forall i.
\end{aligned}
\end{equation}
The resulting minimum error probability $P_{\text{error}} = 1 - P_{\text{guess}}$. Further discussions on optimal measurement strategies are available in~\cite{Barnett09,Bae2015}. In addition, various performance bounds have been studied in~\cite{PhysRevA.81.042329,PhysRevA.110.042401,barnum2002,4578690}.

\subsection{Maximal quantum information leakage}

In this section, we review the definition of maximal quantum leakage~\cite{Farhad2024}, which quantifies information leakage to an arbitrary eavesdropper. Consider a classical random variable $X \in \mathcal{X}$ representing sensitive data, encoded into a quantum system via the ensemble $ \{q_x, \rho_x\}_{x \in \mathcal{X}}$. An adversary, unaware of the realization of $X$, seeks to estimate a potentially randomized function of $X$, denoted by the random variable $Z$. To do so, the adversary performs a measurement described by  POVMs $ \{F_y\}_{y \in \mathcal{Y}}$, obtaining an outcome $Y$ with probability ${P}(Y = y \mid X = x) = \text{Tr}(\rho_x F_y)$, and subsequently forms an estimator $\hat{Z}$. Maximal quantum leakage captures the maximal multiplicative increase in the adversary's guessing probability, optimized over all possible target variables $Z$ and measurement strategies.

\begin{definition}\label{def1}
(Maximal quantum leakage~\cite{Farhad2024}). The maximal quantum leakage from random variable $X$ through the quantum ensemble $ \{q_x, \rho_A^x\}_{x \in \mathcal{X}}$ is defined as
\begin{align}
\mathcal{Q}(X \to A)_{\rho} &\triangleq \sup_{\{F_y\}} \sup_{Z, \hat{Z}} \log \left( \frac{{P}\{Z = \hat{Z}\}}{\max_{z \in \mathcal{Z}} {P}\{Z = z\}} \right) \label{eq:3.1a} \\
&= \max_{\{F_y\}} I_{\infty}^{S}(X; Y) \label{eq:3.1b} \\
&= \max_{\{F_y\}} \log \left( \sum_{y \in \mathcal{Y}} \max_{x \in \mathcal{X}} \text{Tr} (\rho_x F_y) \right), \label{eq:3.1c}
\end{align}
where the maxima are taken over all random variables $Z$ and estimators $\hat{Z}$ with an arbitrary finite support, and all POVMs $\{F_y\}$ with arbitrary finite outcomes $\mathcal{Y}$.
\end{definition}

While Ref.~\cite{Farhad2024} suggests that computing this quantity generally requires iterative algorithms due to the lack of an explicit formula, we demonstrate in this paper that the leakage computation can be formulated as a semidefinite programming (SDP), guaranteeing an optimal solution.

\section{Equivalence of maximal leakage and discrimination, and collective measurement}\label{sec3}

\subsection{Equivalence of leakage and state discrimination}

We first establish the following theorem to demonstrate the equivalence between the maximal guessing probability which is related to information leakage, and the optimal success probability in minimum-error state discrimination.

\begin{theorem}\label{theorem1}
The optimization over arbitrary POVMs followed by optimal classical post-processing is equivalent to the standard minimum-error state discrimination problem:
\begin{equation}\label{eq:main}
    \max_{\{F_y\}} \sum_{y \in \mathcal{Y}} \max_{x \in \mathcal{X}} \left( q_x \mathrm{Tr}(\rho_x F_y) \right)
    = 
    \max_{\{M_{x}\}} \sum_{x \in \mathcal{X}} q_x \mathrm{Tr}(\rho_x M_x),
\end{equation}
where the left-hand side describes a general POVM $\{F_y\}$ with arbitrary outcomes $y$ combined with an optimal classical guessing strategy for $x$, while the right-hand side describes a POVM $\{M_x\}$ with the outcomes directly corresponding to the state indices $x$. Here, $\{q_x\}$ is the prior distribution with $\sum_x q_x =1$.
\end{theorem}

\begin{proof}
First, we show that the left-hand side (LHS) is upper bounded by the right-hand side (RHS). Let $\{F_y\}$ be an arbitrary POVM used in the maximization on the LHS. 
For each measurement outcome $y$, the optimal decision rule is to choose any $x$ maximizing
$q_x\mathrm{Tr}(\rho_x F_y)$, which is equivalently the rule that maximizes the posterior probability of 
$x$ given
$y$. We define this decision rule $\delta(y)$ as
\begin{equation}
    \delta(y) \in  \arg\max_x \left(q_x \mathrm{Tr}(\rho_x F_y) \right).
\end{equation}
Using this decision rule, we construct a new set of operators $\{M_x\}$ by grouping the outcomes $y$ that correspond to the same guess $x$ 
\begin{equation}
M_x \triangleq \sum_{y : \delta(y) = x} F_y.
\end{equation}
Since $\{F_y\}$ is a valid POVM, the constructed set $\{M_x\}$ satisfies $M_x \ge 0$ and $\sum_x M_x = I$, making it a valid candidate for the state discrimination problem on the RHS. By the definition of $\delta(y)$, we can rewrite the LHS as
\begin{equation}
\begin{aligned}
    \sum_{y\in\mathcal Y}\max_{x\in\mathcal X}\left(q_x \mathrm{Tr}(\rho_x F_y)\right)
    &= \sum_y q_{\delta(y)} \mathrm{Tr}(\rho_{\delta(y)} F_y) \\
    &= \sum_x q_x \mathrm{Tr}\Big( \rho_x \sum_{y : \delta(y) = x} F_y \Big)\\
    &= \sum_x q_x \mathrm{Tr}(\rho_x M_x).
\end{aligned}
\end{equation}
This implies that any value achievable on the LHS can be matched by a valid choice of $\{M_x\}$ on the RHS, establishing that $\text{LHS} \le \text{RHS}$.

{Let $\{M_x\}_{x\in\mathcal{X}}$ be an arbitrary POVM feasible for the RHS problem. For each $x\in\mathcal X$, choose a finite family of positive operators $\{F_{x,k}\}_k$ such that
\begin{equation}
\sum_k F_{x,k}=M_x, \quad F_{x,k}\ge 0.
\end{equation}
Treating each pair $(x,k)$ as a measurement outcome $y$, we can construct any arbitrary POVM $\{F_y\}_{y\in\mathcal Y}$. We then have
\begin{equation}\label{eq:lhs_lowerbound}
\begin{aligned}
\!\!\sum_{y\in\mathcal Y}\max_{x'\in\mathcal X}\left(q_{x'}\mathrm{Tr}(\rho_{x'}F_y)\right)\!
&=\!\!\sum_{x\in\mathcal X}\sum_k
\max_{x'\in\mathcal X}\left(q_{x'}\mathrm{Tr}(\rho_{x'}F_{x,k})\right)\\
&\ge
\sum_{x\in\mathcal X}\sum_k q_x\mathrm{Tr}(\rho_xF_{x,k})\\
&=
\sum_{x\in\mathcal X} q_x\mathrm{Tr}\!\left(\rho_x\sum_k F_{x,k}\right)\\
&=
\sum_{x\in\mathcal X} q_x\mathrm{Tr}(\rho_xM_x).
\end{aligned}
\end{equation}
Since the LHS is the maximum over all POVMs $\{F_y\}$, it follows that
\begin{equation}
\text{LHS}\ge \sum_{x\in\mathcal{X}} q_x\mathrm{Tr}(\rho_x M_x).
\end{equation}
As this holds for every feasible POVM $\{M_x\}$ in the RHS problem, taking the maximum over $\{M_x\}$ yields $\text{LHS}\ge \text{RHS}$.
}
Finally, combining both directions, the equality holds.
\end{proof}

Using Theorem \ref{theorem1}, if the prior distribution is uniform, i.e., $q_x=\frac{1}{N}$ for all $x$ where $N \triangleq |\mathcal{X}|$ in this paper, we have
\begin{equation}
    \max_{\{F_y\}} \sum_{y \in \mathcal{Y}} \max_{x \in \mathcal{X}} \left[ \mathrm{Tr}(\rho_x F_y)\right]
    = 
    \max_{\{M_x\}} \sum_{x \in \mathcal{X}} \mathrm{Tr}(\rho_x M_x).
\end{equation}
Then with~\eqref{eq:3.1a}--\eqref{eq:3.1c}, we have
\begin{equation}\label{npguess}
\mathcal{Q}(X \to A)_{\rho}=\log (N P_{\text{guess}} ).
\end{equation}
Therefore, finding the optimal measurement for maximal quantum information leakage is equivalent to finding the optimal measurement for state discrimination with a uniform prior (often considered the worst-case scenario). This problem can be solved efficiently via SDP. While Theorem 2 of \cite{Farhad2024} suggests that the maximal quantum leakage is attained by a POVM with up to $d^2$ rank-1 elements, our Theorem \ref{theorem1} implies that the optimum can always be attained by a POVM with only $\min\{N,d^2\}$ elements.

\subsection{Collective measurement of quantum information leakage}

Let $\mathcal{Q}_n(X \to A)$ denote the maximal quantum leakage for $n$ copies of the quantum states, defined as
\begin{equation}
    \mathcal{Q}_n(X \to A) \triangleq \log \left( \max_{\{F_y\}} \sum_{y\in {\mathcal{Y}}} \max_{x \in \mathcal{X}} \mathrm{Tr}(\rho_x^{\otimes n} F_y) \right),
\end{equation}
where $\{F_y\}_{y\in \mathcal{Y}}$ represents a POVM on the $n$-copy Hilbert space.
Applying Theorem \ref{theorem1}, we can reformulate this as
\begin{equation}
    \mathcal{Q}_n(X \to A) = \log \left( \max_{\{E_x\}} \sum_{x \in \mathcal{X}} \mathrm{Tr}(\rho_x^{\otimes n} E_x) \right),
\end{equation}
where $\{E_x\}_{x \in \mathcal{X}}$ is a POVM acting on the $n$-copy Hilbert space, which allows for entangled measurements.

We then propose the following proposition to characterize $\mathcal{Q}_n(X \to A)$ as a function of  $n$.
\begin{proposition}\label{col1}
    The maximal quantum leakage $\mathcal{Q}_n(X \to A)$ is a {non-decreasing function of $n$} and is bounded by
    \begin{equation}
        \mathcal{Q}_n(X \to A) \leq  \min\{\log(N), \log(d^{2n})\}.
    \end{equation}
\end{proposition}
\begin{proof}
Let $\{\tilde{E}_x^{(n)}\}$ be the optimal POVM that achieves the maximal leakage for $n$ copies, i.e.,
\begin{equation}
    \mathcal{Q}_n(X \to A) = \log \left(\sum_{x \in \mathcal{X}} \mathrm{Tr}\left( \rho_x^{\otimes n} \tilde{E}_x^{(n)} \right)\right).
\end{equation}
We construct a specific POVM $\{E_x^{(n+1)}\}$ acting on $n+1$ copies by appending the identity operator $I$ to the optimal $n$-copy measurement:
\begin{equation}\label{exn}
    \tilde E_x^{(n+1)} \triangleq \tilde{E}_x^{(n)} \otimes I,
\end{equation}
and $\{\tilde E_x^{(n+1)}\}$ forms a POVM.
The quantum leakage for this specific strategy on $n+1$ copies is
\begin{equation}
\begin{aligned}
    &\sum_{x \in \mathcal{X}} \mathrm{Tr}\left( (\rho_x^{\otimes n} \otimes \rho_x) (\tilde{E}_x^{(n)} \otimes I) \right)\\ 
    =& \sum_{x \in \mathcal{X}} \mathrm{Tr}\left( \rho_x^{\otimes n} \tilde{E}_x^{(n)} \right) {\mathrm{Tr}(\rho_x)}
    =\exp (\mathcal{Q}_n(X \to A)).
    \end{aligned}
\end{equation}
Since $\mathcal{Q}_{n+1}$ is defined as the maximization over all possible POVMs on $n+1$ copies, it must be greater than or equal to the value yielded by any specific choice as \eqref{exn}. Therefore, we have
\begin{equation}
    \mathcal{Q}_{n+1} \ge \log \Big(\sum_{x \in \mathcal{X}} \mathrm{Tr}\left( \rho_x^{\otimes (n+1)} \tilde E_x^{(n+1)} \right)\Big) = \mathcal{Q}_n.
\end{equation}
Using Proposition 2 in \cite{Farhad2024} for one-copy scenario, we have $\mathcal{Q}_1(X \to A) \leq \min\{\log(N), \log(d^{2})\}$.
Since the dimension of n-copy state $\rho^{\otimes n}$ is $d^n$, we have $\mathcal{Q}_n(X \to A) \leq \min\{\log(N), \log(d^{2n})\}$.
\end{proof}

Next, we analyze the asymptotic limit using the connection to quantum fidelity. We adopt the standard definition of fidelity between two states
\begin{equation}
    F(\rho_i, \rho_j) = \left( \operatorname{Tr}\sqrt{\sqrt{\rho_i}\rho_j\sqrt{\rho_i}} \right)^2.
\end{equation}
Existing literature~\cite{barnum2002,4578690} provides bounds on the minimum error probability for state discrimination in terms of fidelity
\begin{align}
    P_{\mathrm{error}} &\leq \sum_{i \neq j} \sqrt{p_i p_j} \sqrt{F(\rho_i, \rho_j)}, \label{lower1} \\
    P_{\mathrm{error}} &\geq \frac{1}{2} \sum_{i \neq j} p_i p_j F(\rho_i, \rho_j). \label{lower2}
\end{align}
In the context of collective measurement on $n$ copies, the states become $\rho_i^{\otimes n}$, and the fidelity scales as $F(\rho_i^{\otimes n}, \rho_j^{\otimes n}) = (F(\rho_i, \rho_j))^n$. 
{Let the priors be uniform, $p_i = \frac{1}{N}$ (where $N=|\mathcal{X}|$). Applying \eqref{npguess} and \eqref{lower1}--\eqref{lower2},} we derive the following bounds on $\mathcal{Q}_n(X \to A)$:
\begin{align}
    N - \exp(\mathcal{Q}_n(X \to A)) &\geq \frac{1}{2N} \sum_{i\neq j} \left( F(\rho_i, \rho_j) \right)^{n}, \label{up}\\
    N - \exp(\mathcal{Q}_n(X \to A)) &\leq  \sum_{i \neq j} \left( F(\rho_i, \rho_j) \right)^{n/2}. \label{lb}
\end{align}
For any pair of distinct states, $F(\rho_i, \rho_j) < 1$. Consequently, as $n \to \infty$, the terms on the right-hand side of \eqref{up} and \eqref{lb} vanish. This implies that $\exp(\mathcal{Q}_n(X \to A))$ converges to $N$. Using Proposition \ref{col1}, $\mathcal{Q}_n(X \to A)$ increases monotonically and saturates at $\log(N)$, confirming that perfect distinguishability is unsurprisingly achievable in the asymptotic limit of infinite copies.

\section{Maximal leakage for quantum channels}\label{sec4}
We now extend the proposed framework to analyze the information leakage of quantum channels. Consider a scenario where a classical input $x \in \mathcal{X}$ determines which quantum channel $\mathcal{E}_x$ is applied to a system $A$. This ensemble of operations, denoted by $\{q_x, \mathcal{E}_x\}_{x \in \mathcal{X}}$, represents the quantum encoding of classical data. The objective is to design a universal input state $\rho$ and a measurement strategy $\{F_y\}$ that maximize the adversary's probability of correctly identifying $x$. Following the operational interpretation established in Definition~\ref{def1}, we formalize the maximal quantum leakage for channels as follows.

\begin{definition}
The maximal quantum leakage from a random variable $X$ through the quantum channel ensemble $\{q_x, \mathcal{E}_x\}_{x \in \mathcal{X}}$ is defined as:
\begin{align}
\mathcal{Q}(X \to A)_{\mathcal{E}} & \triangleq \max_{\rho} \max_{\{F_y\}} \sup_{Z, \hat{Z}} \log \left[ \frac{{P}[Z = \hat{Z}]}{\max_{z \in \mathcal{Z}} {P}[Z = z]} \right] \\
&= \max_{\rho} \max_{\{F_y\}} I^{S}_\infty(X; Y) \\
&= \max_{\rho} \max_{\{F_y\}} \log \left( \sum_{y \in \mathcal{Y}} \max_{x \in \mathcal{X}} \mathrm{Tr}(\mathcal{E}_{x}(\rho) F_y) \right),
\end{align}
where the maximization is performed over all input states $\rho$, all POVMs $\{F_y\}$ with arbitrary finite outcomes $\mathcal{Y}$, and all random variables $Z$ and estimators $\hat{Z}$ with finite support.
\end{definition}

Based on the operational equivalence proved in Theorem \ref{theorem1}, the leakage can be reformulated as an optimization problem involving minimum-error discrimination of the output states
\begin{align}\label{epm}
\mathcal{Q}(X \to A)_{\mathcal{E}} &=\max_{\rho} \max_{\{M_x\}}
\log \left(\sum_{x \in \mathcal{X}} \mathrm{Tr}\left(\mathcal{E}_{x}(\rho) M_x \right)\right).
\end{align}

To develop a computational algorithm, it is advantageous to utilize the adjoint channel formulation. Let $V$ denote the objective function within the logarithm
\begin{equation}
  V \triangleq \max_{\rho} \max_{\{M_x\}} \sum_{x \in \mathcal{X}} \mathrm{Tr}(\mathcal{E}_x(\rho) M_x),
\end{equation}
and $\mathcal{Q}(X \to A)_{\mathcal{E}}=\log(V)$.
Recall that the adjoint map $\mathcal{E}^\dagger$ is defined by the relation $\mathrm{Tr}(B \mathcal{E}(A)) = \mathrm{Tr}(\mathcal{E}^\dagger(B) A)$ \cite{watrous2018theory}. Because the variables $\{M_x\}$ and $\rho$ are independent, we can interchange the order of maximization to express $V$ as
\begin{equation}
  V = \max_{\{M_x\}} \max_{\rho} \mathrm{Tr}\left( \rho \sum_{x \in \mathcal{X}} \mathcal{E}_x^\dagger(M_x) \right).
\end{equation}
For any Hermitian operator $R$, the maximization $\max_{\rho} \mathrm{Tr}(\rho R)$ is achieved by projecting onto the eigenspace associated with the largest eigenvalue of $R$~\cite{horn_johnson_2012}, i.e., $\max_{\rho} \mathrm{Tr}(\rho R) = \|R\|_{\infty} = \lambda_{\max}(R)$. Applying this property with $R = \sum_{x} \mathcal{E}_x^\dagger(M_x)$, we obtain
\begin{equation}
\label{eq:optimization_objective}
 V = \max_{\{M_x\}} \left\| \sum_{x \in \mathcal{X}} \mathcal{E}_x^\dagger(M_x) \right\|_\infty.
\end{equation}
The optimization problem in \eqref{eq:optimization_objective} requires finding a POVM $\{M_x\}$ that maximizes the spectral norm of the sum of the adjoint channels. Since this joint optimization is generally non-convex, we propose Algorithm \ref{alg:iterative_leakage} to find a potentially suboptimal by iteratively optimizing over quantum states and POVMs. Because the subproblem in each iteration is convex, our algorithm is guaranteed to converge.

\begin{algorithm}[htbp]
\caption{Iterative joint optimization for maximal quantum channel leakage}
\label{alg:iterative_leakage}
\begin{algorithmic}[1]
\Require Set of quantum channels $\{\mathcal{E}_x\}_{x\in \mathcal{X}}$, convergence tolerance $\epsilon > 0$.
\Ensure Channel leakage $V^*$, input state $\rho^*$, and POVM $\{M_x^*\}$.

\State \textbf{Initialization:}
\State Select a random initial pure state $\rho^{(0)}$.
\State Set iteration counter $k \gets 1$ and initial value $V^{(0)} \gets 0$.

\Repeat
    \State \textit{// Step A: Fix state, optimize measurement (SDP)}
    \State Solve the following SDP for fixed $\rho^{(k-1)}$:
    \State \quad $\{M_x^{(k)}\} \gets \arg \max_{\{M_x\}} \sum_{x} \operatorname{Tr}(\mathcal{E}_x(\rho^{(k-1)}) M_x)$
    \State \quad \quad \quad s.t. $M_x \geq 0 \; \forall x, \sum_{x} M_x = I$.

    \State \textit{// Step B: Fix measurement, optimize state (eigenvalue Problem)}
    \State Compute the adjoint operator sum:
    \State \quad $R^{(k)} \gets \sum_{x} \mathcal{E}_x^\dagger(M_x^{(k)})$.
    \State Compute the largest eigenvalue $\lambda_{\max}$ and eigenvector $|v_{\max}\rangle$ of $R^{(k)}$.
    \State Update state: $\rho^{(k)} \gets |v_{\max}\rangle\langle v_{\max}|$.
    
    \State Update objective value: $V^{(k)} \gets \lambda_{\max}$.
    \State $\delta \gets |V^{(k)} - V^{(k-1)}|$.
    \State $k \gets k + 1$.
\Until{$\delta < \epsilon$}

\State \Return $V^* = V^{(k-1)}$, $\rho^* = \rho^{(k-1)}$, $\{M_x^*\} = \{M_x^{(k-1)}\}$.
\end{algorithmic}
\end{algorithm}

\section{Numerical examples}\label{sec5}

In this section, we present numerical simulations to demonstrate our theoretical findings regarding state leakage under collective measurements and maximal channel leakage. We solved all convex optimization problems using the SDPT3 solver \cite{tutuncu2003solving} within the CVX toolbox \cite{cvx}.

\subsection{Maximal state leakage using collective measurements}

We first investigate the behavior of quantum leakage under collective measurements performed on $n$ identically prepared states. We consider an ensemble consisting of $N=3$ qubit quantum states as
\begin{align}
    |\psi_1\rangle = \begin{pmatrix} 1 \\ 0 \end{pmatrix}, 
    |\psi_2\rangle = \begin{pmatrix} \cos(\frac{\pi}{8}) \\ \sin(\frac{\pi}{8}) \end{pmatrix},  
    |\psi_3\rangle = \begin{pmatrix} \sqrt{0.1} \\ \sqrt{0.9} \end{pmatrix},
\end{align}
with associated density operators $\{\rho_i=|\psi_i\rangle\langle\psi_i|\}_{i=1}^{3}$. We compute the maximal quantum leakage $\mathcal{Q}_n(X \to A)$ for $n \in \{1, \dots, 7\}$. The theoretical upper and lower bounds are derived according to \eqref{up} and \eqref{lb}, respectively.

The simulation results are illustrated in Fig.~\ref{fig:leakage}. It is observed that for $n=1$, the lower bound expression yields a negative value ($3-2 \sum_{i>j} ( F(\rho_i, \rho_j) )^{n/2} < 0$) and thus the corresponding point is omitted from Fig.~\ref{fig:leakage}. The leakage value increases monotonically with $n$, corroborating Proposition~\ref{col1} and confirming that access to multiple copies of the state enhances the adversary's discrimination capability. Furthermore, as $n$ increases, the tensor product states $\rho_i^{\otimes n}$ become increasingly orthogonal. Thus, the leakage asymptotically approaches the limit $\log(3) \approx 1.0986$, indicating that the states become perfectly distinguishable in the limit of large $n$.

\begin{figure}[htbp]
    \centering
    \includegraphics[width=0.45\textwidth]{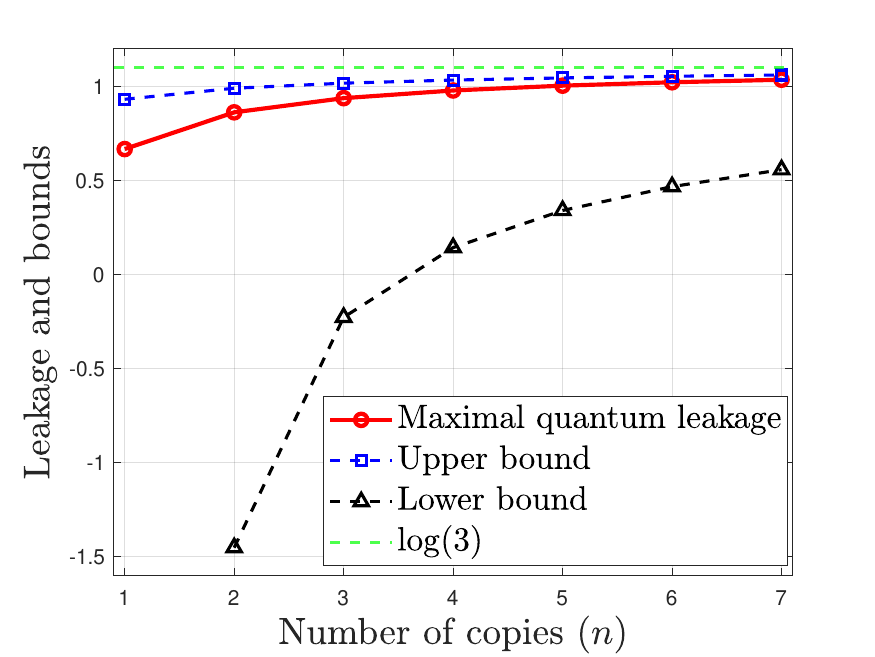}
    \caption{Maximal quantum leakage as a function of the number of copies $n$. The red solid line represents the calculated leakage, bounded by the theoretical limits (blue and black dashed lines) derived in \eqref{up} and \eqref{lb}. The value approaches the  limit $\log(3)$ (green dashed line) as $n$ increases.}
    \label{fig:leakage}
\end{figure}

\subsection{Quantum channel leakage}

Here we evaluate the maximal quantum channel leakage. We consider a set of three unitary channels $\{U_1, U_2, U_3\}$ acting on a single qubit
\begin{align}
    U_1 &= I = \begin{pmatrix} 1 & 0 \\ 0 & 1 \end{pmatrix}, \\
    U_2 &= \exp\left(-i \frac{\pi}{4} \sigma_x\right) = \frac{1}{\sqrt{2}} \begin{pmatrix} 1 & -i \\ -i & 1 \end{pmatrix}, \\
    U_3 &= \exp\left(-i \frac{\pi}{4} \sigma_z\right) = \begin{pmatrix} e^{-i\pi/4} & 0 \\ 0 & e^{i\pi/4} \end{pmatrix},
\end{align}
where $\sigma_x$ and $\sigma_z$ denote the standard Pauli matrices.

We employ Algorithm \ref{alg:iterative_leakage} to compute the channel leakage $\mathcal{Q}(X \to A)_{\mathcal{E}}$. With the convergence tolerance set to $\epsilon=10^{-10}$, the algorithm converges to a channel leakage value of $0.6271$. The convergence behavior with respect to the iteration count is illustrated in Fig.~\ref{fig:without}. To assess the robustness of the algorithm, we performed multiple trials using different random initial pure states. In almost all cases, the algorithm consistently converged to a leakage value of approximately $0.6271$, indicating insensitivity to initialization.

\begin{figure}[htbp]
    \centering \includegraphics[width=0.45\textwidth]{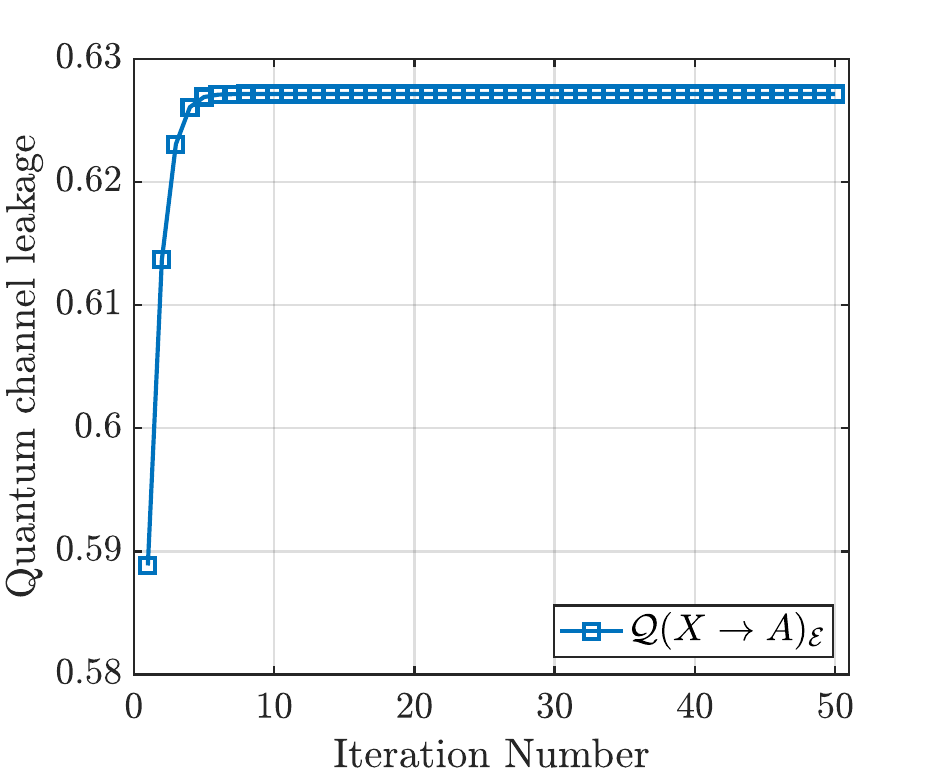}
    \caption{Convergence of Algorithm \ref{alg:iterative_leakage} for quantum channel leakage.}
    \label{fig:without}
\end{figure}

\section{Conclusion}\label{sec6}
In this paper, we have established an operational interpretation of maximal quantum leakage by proving its equivalence to minimum-error quantum state discrimination with equal priors. Our analysis of collective measurements on multiple copies demonstrated that leakage increases monotonically with the number of copies, asymptotically approaching the limit value. Extending this framework to quantum channels, we developed an iterative algorithm for the joint optimization of input states and measurement strategies.  Numerical simulations demonstrated our theoretical findings. Future work will investigate the convergence rate of maximal quantum leakage under collective measurements.

\bibliographystyle{ieeetr}         
\bibliography{leakage} 

\end{document}